\documentclass[10pt]{amsart}
\usepackage[usenames,dvipsnames]{xcolor}
\usepackage{amsmath,amsthm,amssymb,color,comment,csquotes,enumerate,enumitem,fancyhdr,filecontents,graphicx,verbatim}
\usepackage[round]{natbib}

\usepackage{hyperref}
\hypersetup{colorlinks=true,linkcolor=MidnightBlue,citecolor=MidnightBlue,bookmarks=false,backref=page}

\newtheorem{thm}{Theorem}
\theoremstyle{definition}

\newtheorem{defn}[thm]{Definition}

\newcommand{\thistheoremname}{}
\newtheorem{genericthm}[thm]{\thistheoremname}
\newenvironment{namedthm}[1]
  {\renewcommand{\thistheoremname}{#1}%
   \begin{genericthm}}
  {\end{genericthm}}

\newcommand{\bitz}{\mathtt{0}}
\newcommand{\bito}{\mathtt{1}}
\newcommand{\seqx}{\pmb{x}}
\newcommand{\seqy}{\pmb{y}}

\newcommand{\seqemp}{\pmb{\varnothing}}

\newcommand{\pr}{\textsf{p}}

\newcommand{\bitset}{\{\mathtt{0},\mathtt{1}\}}
\newcommand{\bfin}{\{\mathtt{0},\mathtt{1}\}^*}
\newcommand{\binf}{\mathbb{B}^\omega}

\makeatletter
\renewcommand{\@makefnmark}{\hbox{\@textsuperscript{\tiny\@thefnmark}}}
\makeatother

\begin{document}

\title{Solomonoff induction}
\author[Sterkenburg]{Tom F.\ Sterkenburg}
\date{\today. This is a preliminary version. I welcome feedback.}
\address{Munich Center for Mathematical Philosophy (MCMP), LMU Munich \newline \indent Munich Center for Machine Learning (MCML)}
\email{tom.sterkenburg@lmu.de}

\begin{abstract}
This chapter discusses the Solomonoff approach to universal prediction. The crucial ingredient in the approach is the notion of computability, and I present the main idea as an attempt to meet two plausible computability desiderata for a universal predictor. This attempt is unsuccessful, which is shown by a generalization of a diagonalization argument due to Putnam. I then critically discuss purported gains of the approach, in particular it providing a foundation for the methodological principle of Occam's razor, and it serving as a theoretical ideal for the development of machine learning methods.
\end{abstract}

\maketitle

\section{Introduction}

The ``Bayesian establishment view'' on the problem of induction is that the justification for inductive inference is conditional or relative \citep[p.\ 3]{Hut17}. On this subjectivist view, going back to Ramsey, de Finetti, and Savage, rationality dictates coherent degrees of belief and  inference by conditionalization, but puts no constraints on the essential starting point: the prior probability distribution. 
In the context of inductive inference, the prior is an expression of our \emph{inductive assumptions}, which cannot ultimately be justified in any absolute sense. In particular, these assumptions must be restrictive: adequate in some circumstances, but not in others \citep{How00}. In that sense, there is no ``universal'' prior.\footnote{This
	sketch of an establishment view inevitably overlooks many remaining points of contention. For instance, not all authors are as committed as Howson to viewing Bayesianism as a type of logic, including soundness and completeness theorems \citep[ch.\ 7]{How00}; whereas Howson himself does not actually see dynamic consistency (viz., updating via conditionalization) as a core principle of Bayesianism (ibid., p.\ 135), in contrast to \citet[sect.\ 1.4]{Hut17}. Nor is the view itself uncontroversial:  \citet{Can25philimp} criticizes both the historical pedigree and the success of the idea of a conditional justification of induction. For the purpose of this chapter, we do not have to engage with these issues: the relevant point is that there is no universal prior.}

Similarly, a basic lesson in machine learning is that there is no ``free lunch'' \citep{SteGru21syn}. Every learning algorithm must come with restrictive assumptions, or \emph{inductive biases}, which make for successful learning in some possible circumstances but failure of learning in others \citep{ShaBen14,BenSre11talknips}. In that sense, there is no ``universal'' learning method, Bayesian or otherwise.\footnote{Again,
	this sketch overlooks many details, in particular that precise statements of no-free-lunch results require the specification of a learning problem and notion of successful learning, and that we may have free lunches for sufficiently weak such specifications \citep[sect.\ 2.5]{SteGru21syn}.} 

Yet there is one approach in theoretical computer science which promises exactly that. According to its advocates, the proposal by \citet{Sol64ic} 
``may give a rigorous and satisfactory solution'' to the problem of induction, in the form of a Bayesian learner based on a ``universal prior'' \citep[p.\ 354]{LiVit19}.\footnote{In
	similar vein, Solomonoff induction ``arguably solved the century-old induction problem'' by offering a ``compelling theoretical foundation for constructing [\dots]\ an ideal universal prediction system''  \citep[p.\ 1]{GMGHODCRWMAV24icml}.}
	This approach of \emph{Solomonoff induction} has in recent years  again gained some visibility, as it is regularly brought in as a potential piece in the puzzle of understanding modern deep learning \citep{GMGHODCRWMAV24icml,DRDCGMGWAOHV24icml,GolFinRowWil24icml,MinReeValLou25nc,YouWit25inproc,Wil25arxiv}.

This chapter takes a closer look at Solomonoff induction. I start, in section \ref{sec:univpred}, with the crucial idea of defining universal Bayesian prediction methods with the help of the formal constraint of computability.  I point there at a fundamental conflict between two computability desiderata on such  universal predictors, which was already brought out by Putnam in a critique of Carnap's inductive logic. Then, in section \ref{sec:sollev}, I present the approach due to Solomonoff (and \citet{Lev71phd}, who put it on a rigorous mathematical foundation) as a way of trying to overcome this conflict. I discuss a number of reservations about the success of this move, and about the purported (conceptual and practical) gains of the resulting definition. In section \ref{sec:concl}, I offer some concluding remarks on what this suggests for the use of Solomonoff induction in machine learning as  well as for the Bayesian establishment view.\footnote{This
	chapter draws for large part from \citep{Ste18phd}, and  many more historical, mathematical, and philosophical details may be found there. Yet this is not merely a summary of that work: I have updated both structure and emphasis.}

\section{Universal prediction}\label{sec:univpred}

\subsection{Sequential prediction}\label{ssec:seqpred}
The basic inductive learning problem I consider here is the iterated one-step prediction of a developing binary sequence. At each point in time $t$, we have an ordered sequence $\seqx^t = x_1 \dots x_t$ of the outcomes $x_i \in \{\mathtt{0},\mathtt{1}\}$ we have observed so far, and we seek to predict the next outcome  $x_{t+1}$. Our predictions are in the form of probability distributions over the two possible next outcomes. A predictor is a precise rule for making such predictions: formally, it is a function $\pr: \{\mathtt{0},\mathtt{1}\}^* \rightarrow \mathcal{P}$ from the set of all finite binary sequences to all distributions over $\bitset$. I will also write $\pr$ as a two-place function, where $\pr(\seqx,x)$  denotes $\pr$'s predicted probability of next outcome $x$ given observed sequence $\seqx$.\footnote{\label{fn:univsett}There 
	are of course other possible formal learning set-ups one can study, including more restricted ones where predictors issue categorical ($\mathtt{0}$ or $\mathtt{1}$) predictions (\citealp{BluBlu75ic}; also see \citealp{Bel20tcs}), and learning problems which look altogether different from sequential prediction. The current set-up is very general, and to the extent that alternative learning set-ups can be ``coded back'' into it, might even has a claim to itself being truly ``universal'' (cf.\ \citealp[sect.\ I.1]{Ste18phd}).}

\subsection{Universal reliability}\label{sect:univrel}
What would it mean for a predictor to be universal? The central idea in Solomonoff induction is that we can find an answer in the notion of formal computability. 

\subsubsection{Reliability for computability}
We could call a predictor universal, or universally \emph{reliable}, if it eventually converges to the correct predictions on all computable patterns. If $\seqx^\omega$ is an infinite binary sequence, then prediction method $\pr$ converges on $\seqx^\omega$ if its probabilities converge to the correct next bits, i.e., the probabilities for $\mathtt{1}$ converge to 1 on $\mathtt{1}$'s, and to 0 on $\mathtt{0}$'s:
\begin{align}
\pr(\seqx^t,\mathtt{1}) \xrightarrow{t \rightarrow \infty} x_{t+1}.
\end{align}

We could then call a predictor universally reliable if it converges on \emph{all Turing-computable sequences}. 

But we can do something more general still. Namely, we can consider all computable probability measures over the outcome space.\footnote{In
	the following, I skip over several formal details. For the proper definition of (computable) probability measures over Cantor space (the space of infinite binary sequences), see \citep[sect.\ 2.1.1]{Ste18phd}.} 
A probability measure over the infinite binary outcome sequences is induced by a function $\mu$ on all finite sequences which satisfies
\begin{equation}\label{eq:meas} 
\begin{aligned} 
\mu(\seqemp) &=1; \\
\mu(\seqx\mathtt{0})+\mu(\seqx\mathtt{1}) & =\mu(\seqx) \textrm{ for all }\seqx \in \bfin,
\end{aligned} 
\end{equation} 
where $\seqemp$ is the empty sequence and $\mu(\seqx x)$ is the concatenation of $\seqx$ and bit $x$. (We always start with the empty sequence; while the probability of observing either of the two possible one-bit extensions of $\seqx$ is just the probability of observing $\seqx$.) For instance, the  \emph{uniform} measure $\lambda$ is given by $\lambda(\seqx^t)=2^{-t}$ (every sequence of the same length has the same probability).

A \emph{computable} probability measure is simply given by a Turing-computable $\mu$. Now, for given measure $\mu$, we say that predictor $\pr$ is reliable for $\mu$ when, with $\mu$-probability 1, the predictions of $\pr$ converge to $\mu$'s (``true'') probabilities,
\begin{align}
\pr(\seqx^t,x_{t+1}) \xrightarrow{t \rightarrow \infty} \mu(x_{t+1} \mid \seqx^t).
\end{align}

We could then call a predictor universally reliable if it is reliable for all computable measures. Note that this includes the former case of ``deterministic'' convergence on computable sequences, because a computable sequence is given by a computable measure which concentrates all probability on the single sequence.

\subsubsection{A universal  mixture}
This universal reliability is, theoretically, not so hard to achieve. Consider first the following general definition.

\begin{defn}[Bayesian mixture]\label{def:baymix}
Let $\mathcal{H} = \{ \mu_i \}_{i \in \mathbb{N}}$ be a countable class of measures. Let \emph{weight function} $w: \mathbb{N} \rightarrow (0,1]$ be such that $\sum_{i \in \mathbb{N}} w(i) \leq 1$. The Bayesian mixture measure for $w$ over $\mathcal{H}$ is defined by 
\begin{align}
\xi^\mathcal{H}_w (\seqx) := \sum_{i \in \mathbb{N}} w(i) \mu_i(\seqx), 
\end{align}
and in turn defines a mixture predictor by\footnote{In
	general, there is a formal correspondence between predictors and measures, or at least, measures that assign positive probability to each finite sequence (see \citealp[sect.\ 3.1.1]{Ste18phd}). \label{fn:predsmeass}} 
\begin{align}
\pr^\mathcal{H}_w (\seqx,x) := \xi^\mathcal{H}_w(x \mid \seqx) = \frac{\xi^\mathcal{H}_w (\seqx x)}{\xi^\mathcal{H}_w (\seqx)}. 
\end{align}
\end{defn}

We have the following general reliability result.

\begin{thm}[Bayesian consistency]\label{thrm:bayconsist}
For mixture predictor \emph{$\pr^\mathcal{H}_w$}, for all $\mu \in \mathcal{H}$, with $\mu$-probability 1, 
$\pr^\mathcal{H}_w(\seqx^t,x_{t+1}) \xrightarrow{t \rightarrow \infty} \mu(x_{t+1}\mid \seqx^t).$
\end{thm}
\begin{proof}
This follows directly by the Blackwell-Dubins theorem \citeyearpar{BlaDub62ams} from the fact that every $\mu \in \mathcal{H}$ is \emph{dominated} by the mixture $\xi_w$ (meaning, for each event $A$, if $\mu(A)>0$ then also $\xi(A) > 0$).
\end{proof}

Now note that, since there are  countably many Turing machines, the class $\mathcal{H}^\mathrm{comp}$ of all computable measures is countable. So there exists an enumeration $\{ \mu_i \}_{i \in \mathbb{N}}$ of all computable measures, and we can define a corresponding mixture measure $\xi^\mathrm{comp}_w$ and predictor $\pr^\mathrm{comp}_w$, which, by Bayesian consistency, is reliable for all computable measures. A universally reliable prediction method: what is not to like?\footnote{One 
	might here already worry that the class of all computable measures, while countable,  is not actually enumerable in an effective manner. I come to this in sections \ref{ssec:diagarg}--\ref{ssec:semisemis} below.}

\subsection{Universal optimality}\label{sec:univoptcomp}
One philosophical concern is the identification of universal reliability with reliability for computability, that is, reliability for all the ways the sequence under investigation could be \emph{computably} generated. While certainly general, this is still a restriction, an inductive assumption (cf.\ \citealp[p.\ 77]{How00}). To maintain that this inductive assumption is truly universal is to commit to a stance that we will or can only ever encounter sequences that are computably generated, perhaps a ``physical Church-Turing thesis'' that nature is computable (see \citealp{CopSha20inc}); a stance that is hard to evaluate.  

\subsubsection{Optimality among all methods}
The statistician \citet[p.\ 1274]{Sch85aos} quips that ``Nature is not (to my knowledge) hampered by the same computability restrictions as statisticians are.''
But \emph{we}, statisticians, or generally developers of prediction methods, are plausibly bounded by computability. Or at least, the prediction methods we can develop are. 

Following this idea, we could explicitly restrict the possible prediction methods to prediction \emph{algorithms}, meaning (accepting the original Church-Turing thesis)  the Turing-computable predictors. We could then say that a predictor is universal, or universally \emph{optimal}, if it is in some sense as reliable as any prediction algorithm, any computable predictor.

This move bears a resemblance to Reichenbach's attempted vindication of induction (see, e.g., \citealp{Sal91e}). His idea, in our terms, was that we must accept that ``the inductive method'' could not be shown to be universally reliable, because we cannot a priori exclude the possibility of  outcome sequences on which the method would fail. Similarly, we cannot a priori exclude the possibility of incomputable outcome sequences, on which our universal mixture would fail. However, we \emph{may} be able to show that induction is successful, \emph{in those cases in which success is attainable at all}.

 Reichenbach's own proposal, in terms of the ``straight rule of induction'' and its success whenever there is a limiting relative frequency, was ultimately not convincing \citep{Gal11syn}. But perhaps the current proposal, that we might have a universally optimal method which is successful whenever \emph{some} method is (whenever successs is attainable at all), is more promising.

\subsubsection{Aggregating predictors}
And indeed, we can show something quite strong about the optimality of our  mixture predictor $\pr_w^\mathrm{comp}$. Note first that, in the same way that a mixture measure gives a mixture predictor (def.\ \ref{def:baymix} above), any measure $\mu$ gives a corresponding predictor $\pr$. Thus a class of measures corresponds to a pool of predictors.\footnote{Again
	(fn.\ \ref{fn:predsmeass}), this strictly speaking only holds for measures which are everywhere positive, for a corresponding predictor would ``go into a coma'' \citep[p.\ 305]{Kel96} on a probability-0 outcome. This is not a real issue here, since a predictor would simply drop out of the mixture over the pool when this happens to it (provided this does not happen to all predictors---which is excluded in the relevant case  of all computable elements).}  
	A mixture predictor $\pr^\mathcal{H}_w$ can also  be written directly as a predictor which \emph{aggregates} all the predictions of the corresponding pool of predictors:\footnote{See
	\citep[sect.\ 3.2.2]{Ste18phd}.}
\begin{align}
\pr^\mathcal{H}_w(\seqx,x) = \sum_{i \in \mathbb{N}} w_{t+1}(i) \pr_i(\seqx,x).
\end{align}
Here $\pr_i$ is the predictor corresponding to $\mu_i$, and the weight function is updated by 
\begin{align}
w_{t+1}(i) = \frac{w_t(i)\pr(\seqx_{t-1},x_t)}{Z},
\end{align}
with normalization factor $Z = \sum_{i \in \mathbb{N}} w_t(i)\pr(\seqx_{t-1},x_t)$ and initial weights $w_0=w$.

\subsubsection{Optimality}
Now one might hope that we could make the optimality of an aggregating predictors precise as follows. For each possible infinite outcome sequence, if \emph{some} predictor in the pool converges on the sequence, then the aggregator does so, too. Unfortunately, this is too strong.\footnote{Indeed,
	for the Solomonoff-Levin semi-predictors which I will discuss in section \ref{sec:sollev} below, and which are aggregators for a pool which contains themselves, it can be shown that for each pair of such semi-predictors there are sequences on which they do not converge on the same predictions \citep[thrm.\ 3.2]{Ste18phd}. See \citep{Nie20bjps} for a study of the ``deterministic'' convergence of (Bayesian) learners on all events.\label{fn:detconv} }\textsuperscript{,}\footnote{Alternatively, 
	it is possible to give a reinterpretation of the consistency theorem \ref{thrm:bayconsist} as a \emph{merger-of-opinion} result (see, e.g., \citealp{Hut15rsl}), saying that each predictor anticipates with certainty that an aggregator's predictions convergence to its own. I will not pursue this route, also because it becomes harder to make sense of this interpretation for the semi-predictors which I have to turn to in section \ref{sec:sollev} below (cf.\ \citealp[sect.\ 3.3.2.1]{Ste18phd}).}
	
We will instead take a different route, to at least get quite close to this idea.\footnote{This
	is an approach more broadly studied in \emph{prediction with expert advice} or \emph{competitive online learning} \citep{CesLug06}. Results from this field are at the core of the \emph{meta-inductive} approach to the problem of induction (\citealp{Sch19}; also see \citealp{Ste20epi}).} 
	For this route we need to 
introduce a \emph{loss function} $\ell_{\pr}$ which quantifies the cost or loss of a prediction by $\pr$  in light of the actual outcome $x$. A standard loss function in sequential prediction is the logarithmic or simply \emph{log} loss function, defined by
\begin{align}
\ell_{\pr}(\seqx,x):= -\log \pr(\seqx,x).
\end{align}
The \emph{cumulative} loss of predictor $\pr$ on sequence $\seqx$ is the sum of instantaneous losses of $\pr$'s predictions in the course of the generation of $\seqx$:
\begin{equation}
L_\pr(\seqx^s) := \sum_{t=0}^{s-1} \ell_{\pr}(\seqx^t,x_{t+1}).
\end{equation}
Finally, the cumulative \emph{regret}  of one predictor $\pr_1$ relative to another predictor $\pr_2$ on sequence $\seqx$ is the surplus loss that $\pr_1$ incurred in comparison to $\pr_2$,
\begin{equation}
R_{\pr_1,\pr_2}(\seqx) := L_{\pr_1}(\seqx) - L_{\pr_2}(\seqx).
\end{equation}

For the log loss, we can write the cumulative loss of predictor $\pr$, with  $\mu$ the corresponding measure, as
\begin{align}
L_{\pr}(\seqx^s) &= \sum_{t=0}^{s-1} -\log \mu(x_{t+1}\mid \seqx^t) = -\log \prod_{t=0}^{s-1} \mu_\pr(x_{t+1}\mid \seqx^t) = -\log \mu_\pr(\seqx^s),
\end{align} 
and hence the cumulative regret of $\pr_1$ relative to $\pr_2$, with corresponding $\mu_1$ and $\mu_2$, 
\begin{align}
R_{\pr_1,\pr_2}(\seqx) =  -\log \mu_{1}(\seqx) - \left( -\log \mu_{2}(\seqx) \right) = -\log \frac{\mu_{1}(\seqx)}{\mu_{2}(\seqx)}.
\end{align}

Let us write $R^\mathcal{H}_{w,i}$ for the regret function for the aggregating predictor $\pr^\mathcal{H}_w$ relative to predictor $\pr_i$ corresponding to $\mu_i \in \mathcal{H}$. Then, for the log loss, we can derive the following \emph{constant} bound on the aggregator's regret relative to any predictor in the pool, on \emph{each} outcome sequence. 

\begin{thm}[Optimality]\label{thrm:mixlogloss}
For aggregating predictor \emph{$\pr^\mathcal{H}_w$}, every \emph{$\pr_i$} in the pool corresponding to $\mathcal{H}$, and every finite sequence $\seqx$,
$$R^\mathcal{H}_{w,i}(\seqx) \leq -\log w(i).$$
\end{thm}
\begin{proof}
We have
\begin{align*}
R^\mathcal{H}_{w,i}(\seqx) =  -\log \frac{\xi_w(\seqx)}{\mu_i(\seqx)} 
\leq -\log \frac{w(i) \mu_i(\seqx)}{\mu_i(\seqx)} 
= -\log w(i) 
\end{align*}
because of the domination property $\xi_w (\seqx) = \sum_{j \in \mathbb{N}} w(j) \mu_j (\seqx) \geq w(i) \mu_i(\seqx)$.
\end{proof}

\subsubsection{The universal aggregator}
So in particular, for our universal aggregating predictor over the pool of all computable predictors, we have that the universal aggregator's log regret relative to any given predictor in the pool is bounded by a constant, on \emph{any}  outcome sequence, That is, accepting again the identification of predictors with prediction algorithms, with \emph{computable} predictors, we can say that the aggregating predictor is provably successful (has low loss) \emph{whenever any predictor is successful} (has low loss). So there we have it: a truly universally optimal prediction method.

\subsection{The diagonal argument}\label{ssec:diagarg}
Except that if we restrict the possible predictors to the computable predictors, as we did, \emph{then the aggregating predictor is itself no longer a possible predictor}.

This observation goes back to an argument of \citet{Put63inc1} against Carnap's \citeyearpar{Car50,Car52,CarSte59} program of inductive logic. Putnam took Carnap to be likewise after something like a universal formal method for inductive inference, and sought to demonstrate once and for all the futility of this endeavor  by giving a proof that two plausible conditions on such a universal method are mutually exclusive.\footnote{For 
	more details, see \citep[ch.\ 1]{Ste18phd}.}

Putnam's first condition was that, as a rational reconstruction of our inductive practices, Carnap's inductive method should at least be able to eventually detect all computable patterns. This could be interpreted, in our terms, as a reliability condition, but also as an optimality condition, insofar the idea is that the method should pick up any pattern any of our possible inductive practices could. I will simply state the condition as follows, for both interpretations.

\setcounter{enumi}{1}
\begin{enumerate}[label=(\Roman*)]
\item\label{cond1star} $\pr$ is reliable (optimal) for computability.
\end{enumerate}
His second condition was precisely that, in order to still count as an actual method at all, it should satisfy some condition of computability. I simply state it as:
\begin{enumerate}[label=(\Roman*)]
\setcounter{enumi}{1}
\item\label{cond2star} $\pr$ is computable.
\end{enumerate}

Putnam's original diagonalization proof, based on G\"odel's, was somewhat involved, but for our two conditions the argument is straightforward (also see \citealp[p.\ 702]{Kel16inc}).

\begin{thm}[Putnam's diagonalization]\label{thrm:putdiag}
There can be no \emph{$\pr$} satisfying \emph{\ref{cond1star}} and \emph{\ref{cond2star}}.
\end{thm}
\begin{proof}
Suppose there is such a $\pr$. We construct a computable sequence $\seqx^\omega$ on which $\pr$ will never converge (on which it will have unbounded log loss and hence log regret relative to the computable predictor corresponding to $\seqx^\omega$), as follows. For each $t$, compute and pick the (a) $x_{t+1}$ for which $\pr(\seqx^t,x_{t+1}) \leq 0.5. $
\end{proof}

\section{The Solomonoff-Levin approach}\label{sec:sollev}
Solomonoff, in his landmark paper \citeyearpar{Sol64ic}, puts forward a number of ``models'' for prediction. One of those is precisely our mixture over computable measures, and he also notes this mixture is itself not computable (ibid., p.\ 21). His other proposed models use a different idea, and give a different kind of prediction method. However, Solomonoff's own presentation is not fully rigorous, and the appropriate formal framework was only made formally precise by Kolomogorov's student Levin \citeyearpar{ZvoLev70rms,Lev71phd}. It is this framework which I will now develop first.\footnote{This
	presentation draws from \citep[ch.\ 2]{Ste18phd}, which contains more details. For other  presentations of this framework, see \citep{SheUspVer17,LiVit19}.}   

\subsection{Semi-computable semi-measures}\label{ssec:semisemis}
It is instructive to view the motivation for this framework as an explicit attempt to avoid diagonalization, and to first look at the analogous move in the context of computable functions on the natural numbers.

\subsubsection{Computable functions}
It is an elementary fact of computability theory that the class of all \emph{total computable} (t.c.)\ functions $f: \mathbb{N} \rightarrow \mathbb{N}$ is diagonalizable (see, e.g., \citealp{Soa16}). 
 More precisely: if there were to exist a \emph{universal} t.c.\ function $\mathring{f}$ that can emulate every other t.c.\ function (meaning that $\mathring{f}(i,x)=f_i(x)$ for a listing $\{f_i\}_{i \in \mathbb{N}}$ of all t.c.\ functions), then we could infer a \emph{diagonal function} $g$ (say $g(x):=\mathring{f}(x,x)+1$) that is t.c.\ yet distinct from every single $f_i$ (because $g(i)=f_i(i)+1\neq f_i(i)$ for all $i$), which is a contradiction. 
  This also means that there can be no listing $\{f_i\}_{i \in \mathbb{N}}$ of all t.c.\ functions that is itself computable: the class of t.c.\ functions is not computably enumerable.

\subsubsection{Partial computable functions}
However, the standard model of the Turing machine, which characterizes the Turing-computable (we simply say: computable) functions, is not restricted to total functions. Namely, it is possible that a Turing machine fails to halt on a particular input, so that the function it computes, for this input, is undefined. 

The Turing machines therefore correspond to the larger class of \emph{partial computable} (p.c.)\ functions $\varphi$. This expansion defeats the foregoing diagonalization argument (consider: what if $\varphi_i(i)$ is undefined?). Indeed, the class of all p.c.\ functions \emph{is} computably enumerable, which can be understood most directly from the observation that we can effectively encode into integer indices all possible valid (descriptions of instructions sets for) Turing machines. In particular, we can define a universal Turing machine (or indeed infinitely many such machines, for each such encoding). Namely, given a computable enumeration $\{ \varphi_i\}_{n \in \mathbb{N}}$ of all p.c.\ functions, we can design a machine which takes an index $i$ and input $x$, recovers $\varphi_i$, and computes $\phi_i(x)$. Such a machine computes a universal p.c.\ function $\mathring{\varphi}: i,x \mapsto \varphi_i(x)$.

%
%
%

\subsubsection{Computable measures}
We return to the probability measures over Cantor space, our outcome space of infinite binary sequences. Let a \emph{transformation} $\lambda_F$ of the uniform measure by Borel function $F: \binf \rightarrow \binf$ be defined by $\lambda_F(A) = \lambda(F^{-1}(A))$ for each measurable event $A$. Every Borel measure $\mu$ on Cantor space can be obtained as a transformation of $\lambda$ by some Borel function. 

In order to now impose the restriction of computability on this characterization of measures, we need to express these transformations in terms of functions on \emph{finite} sequences. Skipping various details, we are led to computable functions that have a certain property of monotonicity, which are precisely given by a type of Turing machines that are known as \emph{monotone} machines. 

One can visualize a monotone machine as operating on a steady stream of input symbols, producing a
 stream of output bits in the process. Say that $\seqx$ is an $M$-\emph{description} for $\seqy$ if monotone $M$, when fed $\seqx$ for input, produces a sequence with initial segment $\seqy$. The set $D_M(\seqy)$ of \emph{minimal} $M$-descriptions for $\seqy$ are all the sequences $\seqx$ such that $\seqx$, but no initial segment of $\seqx$, is an $M$-description for $\seqy$.

A transformation $\lambda_M$ of the uniform measure by the monotone Turing machine $M$ is then defined by 
\begin{align}\label{eq:trans}
\lambda_M(\seqy) := \sum_{\seqx \in D_M(\seqy)} \lambda(\seqx).
\end{align}

The computable measures are given by the transformations by  monotone machines $M$ that are ``almost total.'' This means that for $\lambda$-almost all infinite sequences $\seqx^\omega$ (i.e., with $\lambda$-probability 1) $M$ will, for increasingly long initial segments of $\seqx^\omega$, return an increasingly long output sequence, thus producing an infinite $\seqy^\omega$.

%

\subsubsection{Semi-computable semi-measures}
In general, a monotone machine can fail to produce such an unending sequence. It can be partial in the sense that for some $\seqx$, and for all extensions of $\seqx$, it produces a finite sequence $\seqy$ and no more than that $\seqy$. This translates into the possibility that, in the definition of $\lambda_M$, measure for extensions of $\seqy$ ``gets lost,'' and $\lambda_M(\seqy)$ is strictly greater than $\lambda_M(\seqy\bitz)+\lambda_M(\seqy\bito)$.

That is to say, the transformations $\lambda_M$ for monotone machines $M$ are \emph{semi-measures} $\nu$, which only satisfy the inequalities
\begin{equation}\label{eq:semimeas} 
\begin{aligned} 
\nu(\seqemp) &\leq1; \\
\nu(\seqx\mathtt{0})+\nu(\seqx\mathtt{1}) &\leq \mu(\seqx) \textrm{ for all }\seqx \in \bfin.
\end{aligned} 
\end{equation}

These semi-measures moreover satisfy a weaker notion of computability. Namely, the transformations $\lambda_M$ for all $M$ are precisely the (\emph{lower}) \emph{semi-computable} semi-measures, which are the functions $\nu$ satisfying \eqref{eq:semimeas} and such that there is a computable $f_\nu$ approximating $\nu$ from below. That is, for all $\seqx \in \bfin$, all $s \in \mathbb{N}$,
\begin{equation}\label{eq:semimeas} 
\begin{aligned} 
  f_\nu(\seqx,s) &\leq  f_\nu(\seqx,s+1); \\
\lim_{s\rightarrow \infty} f_\nu(\seqx,s) &= \nu(\seqx).
\end{aligned} 
\end{equation}


\subsubsection{Universal semi-computable semi-measures}
Just like the usual Turing machines, the class of all monotone machines can be computably enumerated, and we can define \emph{universal} monotone machines $U$. This in turn allows us to define \emph{universal} transformations $U$. Such a universal transformation for enumeration $\{ M_i \}_{i \in \mathbb{N}}$ of all monotone machines can be written as a mixture of all the transformations,
\begin{align}
\lambda_U(\seqy) = \sum_{i \in \mathbb{N}} \lambda(\seqx_i) \lambda_{M_i}(\seqx),
\end{align}
with $\{ \seqx_i \}_i$ some encoding of the integers. From this it follows directly that a universal transformation is a \emph{universal} semi-computable semi-measure, that is, a measure $\mathring{\nu}$ satisfying, for all $\nu$, for some constant $c_\nu \in (0,1)$, 
\begin{align*}
\mathring{\nu}(\seqx) \geq c_\nu \cdot \nu(\seqx) \textrm{ for all }\seqx \in \bfin.
\end{align*}

\subsubsection{The representation theorem}
There are in fact a number of different ways of  defining the universal transformations, of characterizing the same class $\mathcal{SL}$ of universal semi-computable semi-measures.

First of all, instead of the uniform measure $\lambda$, we could have used any other continuous computable measure $\mu$ to define the universal transformations.\footnote{A
	measure is continuous if it does not have any \emph{atoms}: single infinite sequences which are assigned positive probability.}

Second of all, we could have directly used the definition of a Bayesian mixture $\xi^\mathrm{semi}_w$ over the class $\mathcal{H}^\mathrm{semi}$ of all semi-computable semi-measures.

Following good philosophical tradition \citep{Sup02}, we can collect and summarize these observations, linking ostensibly very different representations of the same objects, under the header of a  \emph{representation theorem}. 

\begin{namedthm}{Representation Theorem}[\citealp{WooSunHut13incoll,Ste17tocs}]\label{thrm:repr}
For all continuous computable measures $\mu$, and acceptable enumerations $\{\nu_i\}_i$ of $\mathcal{H}^\mathrm{semi}$ to define $\xi_w^\mathrm{semi}$,
\begin{align}
\mathcal{SL} = \{ \mu_U \}_U = \{ \xi^\mathrm{semi}_w\}_w,
\end{align}
with $U$ ranging over the universal monotone machines compatible with $\mu$ and $w$ ranging over the computable weight functions.\footnote{For
	the definition of \emph{acceptable enumeration} of the semi-computable semi-measures, which matches the standard definition of acceptable enumeration of the p.c.\ functions \citep[p.\ 21]{Soa16}, see \citep[p.\ 194]{Ste18phd}. The requirement that the machines be \emph{compatible} with $\mu$ has to do with the need to exclude that encodings used by $U$ for the machines receive non-zero probability from $\mu$ (ibid., p.\ 190).} 
\end{namedthm}
\begin{proof}
This needs some space. See \citep[sect.\ 2.2]{Ste18phd}.
\end{proof}

\subsubsection{Universality included}
Theorem \ref{thrm:putdiag} showed that we can diagonalize the class of computable measures: any measure which dominates all computable measures must itself fall outside of the class. By expanding the class of computable measures to the class of semi-computable semi-measures, we have escaped diagonalization. Namely, this class contains universal elements, which dominate all elements in the class, but which are still themselves in the class.


\subsection{The Solomonoff-Levin semi-predictors}
The elements of class $\mathcal{SL}$ are the Solomonoff-Levin semi-measures. The Solomonoff-Levin semi-predictors are simply the semi-predictors corresponding---as before, by conditionalization---to the Solomonoff-Levin semi-measures.

\


%
%
%

\subsection{Diagonalization strikes again}
Could we say that the Solomonoff-Levin semi-predictors are universally optimal predictors?

To retrace the same kind of reasoning as in section \ref{sec:univoptcomp}, which started with the identification of the possible predictors with those predictors that are computable, we would now have to weaken the required level of computability, and still allow  as possible predictors elements which are only semi-computable. The hope would be that we could then unite versions of Putnam's two requirements on a universal predictor: a predictor that is universal for the class of possible predictors, while still being a legitimate predictor itself.

On a first glance, it seems that with the Solomonoff-Levin semi-predictors we have exactly that. It follows directly from the fact that the Solomonoff-Levin semi-predictors are aggregating predictors for the semi-predictors corresponding to the semi-computable semi-measures that the Solomonoff-Levin semi-predictors are optimal among the semi-predictors corresponding to the semi-computable semi-measures. Moreover, the Solomonoff-Levin semi-predictors are themselves semi-predictors corresponding to semi-computable semi-measures.

There is, however, a catch. The weakened desideratum of computability---semi-computability---applies to the underlying probability measures. It does not necessarily carry over to the corresponding semi-predictors. Indeed, the Solomonoff-Levin semi-\emph{predictors} are no longer semi-computable.\footnote{This
	was first shown by \citep[thrm.\ 6]{LeiHut15alt}. The following more straightforward proof is from \citep{Ste19erk}.} 

\begin{thm}
The Solomonoff-Levin semi-predictors are not semi-computable.
\end{thm}
\begin{proof}
Suppose that $\pr^\mathrm{SL}_w(\seqx,x) := \xi^\mathrm{semi}_w(x \mid \seqx)$, for some $\xi_w^\mathrm{semi} \in \mathcal{SL}$, is semi-computable. 
We can now construct a computable infinite sequence $\seqx^\omega$ as follows. Start calculating $\pr_\mathrm{SL}(\bito^t,\bito)$ from below in dovetailing fashion for increasing $t \in \mathbb{N}$, until a $t_0$ such that $\pr_\mathrm{SL}(\bito^{t_0},\bito )>0.5$ is found (since $\bito^\omega$ is obviously computable, and $\pr_\mathrm{SL}$ is in particular reliable for all computable sequences, such $t_0$ must exist). Next, calculate $\pr_\mathrm{SL}(\bito^{t_0}\bitz\bito^t,\bito)$ for increasing $t$ until a $t_1$ with $\pr_\mathrm{SL}(\bito^{t_0}\bitz\bito^{t_1},\bito)>0.5$ is found (again, $t_1$ must exist because $\bito^{t_0}\bitz\bito^\omega$ is computable). Continuing like this, we obtain a list $t_0,t_1,t_2,...$ of positions; let $\seqx^\omega := \bito^{t_0}\bitz\bito^{t_1}\bitz\bito^{t_2}1 \dots $. Sequence $\seqx^\omega$ is computable, but by construction the predictions of $\pr^\mathrm{SL}_w(\seqx,x)$ will not converge on this sequence, contradicting the predictor's reliability for computable sequences.
\end{proof}

%
%
%
%

The Solomonoff-Levin semi-predictors are merely \emph{limit-computable}: there is a computable approximation function $f(\seqx,x,s)$ which for each $\seqx,x$ converges to the predicted probability for $x$ given $\seqx$. The difference with semi-computability is that there we could be sure that the approximation at $s+1$ is more accurate than the one at $s$; with limit-computability we lose even that. 

Could we nevertheless not try to identify the possible prediction methods with all limit-computable predictors and derive a universally optimal predictor for that class? No, we cannot: at this level we can repeat the original diagonal proof of theorem \ref{thrm:putdiag}, with the Halting set as an oracle. In general, with so-called ``relativizations'' (computability-theoretic terminology for using incomputable sets as oracles) the two diagonal proofs given are enough to show that this will not work at any level in the arithmetical hierarchy of computability levels.

\subsection{What have we gained?}
So the best we can say at this point is that we have universally optimal prediction methods, the Solomonoff-Levin semi-predictors, if we identify the possible prediction methods with those semi-predictors which correspond to semi-computable semi-measures.

\subsubsection{The possible prediction methods}
The obvious discomfort is  that this hardly looks like a natural characterization of possible prediction methods. 

The very notions of semi-measure and semi-predictor already raise, to say the least, questions of interpretation. There have been attempts in the literature to offer such interpretations (\citealp{Cam13incoll,MarEveHut16inproc,WyeHut26inproc}), but in the end, in line with how we introduced it here, the notion of semi-measure is primarily a formal necessity to arrive at a weaker notion of computability, semi-computability, which allows for universal elements.\footnote{Interpretations
	that the ``probability leak'' in semi-measures represents something like the probability of the end of the universe are in tension with the idea that our framework of sect.\ \ref{ssec:seqpred} is universal because we can encode anything back into a binary prediction problem (fn.\ \ref{fn:univsett} above): why not encode this possibility in the binary sequence?}\textsuperscript{,}\footnote{Solomonoff
	also described a \emph{normalization} of the Solomonoff-Levin semi-measure, resulting again in a proper measure \citep[sect.\ 4.5.3]{LiVit19}. This measure is merely limit-computable \citep{LeiHut15alt}, which raises the question why we would not be happy with a straightforward mixture over all limit-computable measures.}

The notion of semi-computability (or indeed limit-computability), as an outer constraint on possible predictors, might look plausible insofar such predictors are still effectively ``within reach'': we can still approximate them by computable means. But consider what using a strictly semi- or limit-computable semi-predictor would actually look like (cf.\ \citealp[p.\ 104]{KelJuhGly94inc}, on p.c.\ predictors). 

In order to actually issue predictions in any given round, we cannot wait forever for increasingly accurate approximations of our method's predictions. So we have to make a decision when to go with the current approximation. Now either these decisions can be formalized in a computable procedure, in which case the resulting semi-predictor cannot actually be universally optimal or reliable anymore. Or our decisions are somehow incomputable, putting the lie to the idea that computable approximability is still somehow within effective reach.\footnote{I
	return to the idea that we can at least design computable approximations to the Solomonoff-Levin semi-predictors in sections \ref{sssec:approxs}--\ref{sssec:approxsskep} below.}

\subsubsection{Back to reliability} 
We could of course retreat from our optimality interpretation altogether, upholding only the second of Putnam's conditions. Namely, the Solomonoff-Levin semi-predictor, while incomputable, is still universally reliable: it converges on all computable sequences (indeed, computable measures).

But if we drop any computability requirements on a universal prediction method, then it is not clear why we would not be satisfied with the much more straightforward definition of section \ref{sect:univrel}.  Directly defined from a mixture measure over all computable measures, this also gives an actual predictor instead of a semi-predictor. Moreover, the inductive assumption as encoded in this predictor, that the data is generated from a computable measure, is seemingly more natural than the inductive assumption of the Solomonoff-Levin semi-predictors, that the data is generated from some semi-computable semi-measure.\footnote{Perhaps
	the most natural interpretation of the latter is that we assume that the data is generated by some monotone machine on random input (\citealp[p.\ 357]{LiVit19}; \citealp[p.\ 2]{GMGHODCRWMAV24icml}). This does introduce the possibility that at some finite point no further outcomes arrive, which changes the learning problem as stated in section \ref{ssec:seqpred}. Moreover, it is not fully clear how to make sense of a predictor's convergence on a semi-measure. Also see \citep[sects.\ 3.2.4.3, B.2.2.4]{Ste18phd}.} 

\subsubsection{The sweet spot?}
Perhaps we could still maintain that the Solomonoff-Levin definition at least strikes the best possible trade-off between Putnam's two incompatible conditions. The prediction methods are not computable, but it comes close; their reliability or optimality is a bit too widely cast, but not overly so.

But then, of course, one could just as well say that the definition is unsatisfying on both counts, because it undershoots the one requirement and overshoots the other. Indeed, if the first is supposed to be a \emph{requirement}, then the conclusion of our discussion so far should not be that we can give a theoretical definition of universal prediction methods that at least comes quite close to what we want. The conclusion should simply be that we cannot get there: that there can be no universal prediction method.

\subsection{Still, what have we gained?}
Let us for a moment disregard the previous concerns and  accept the Solomonoff-Levin semi-predictors as bona fide universal (universally reliable) predictors. What can we learn from the Solomonoff-Levin definition and its properties, especially with an eye to actual machine learning?

\subsubsection{Occam's razor}
Much has been made of the association of Solomonoff's original definition with data compression and Occam's razor, the methodological principle to prefer simplicity in inductive inference.

In the definition $\eqref{eq:trans}$ of a transformation of the uniform measure $\lambda$ via a monotone machine $M$, we see that $\lambda_M(\seqx)$ is higher as the $M$-descriptions $\seqy \in D_M(\seqx)$ have higher uniform probability $\lambda(\seqy)=2^{-|\seqy|}$, that is, as the $M$-descriptions are \emph{shorter}. This gives rise to the interpretation that the probability $\lambda_M(\seqx)$ is higher as $\seqx$ is more \emph{compressible} by machine $M$.\footnote{More 				generally, there exists a formal correspondence between description length functions and probability distributions \citep{Gru07,SheUspVer17}, where shorter descriptions correspond to higher probabilities. This, importantly, highlights that talk about ``compressibility'' and ``simplicity'' is relative to some particular description length function in the exact same way that ``high probability'' is relative to some particular probability distribution.\label{fn:codesprobs}   }

Moreover, if we use a universal monotone machine $U$, then the corresponding transformation $\lambda_U$---a Solomonoff-Levin semi-measure---arguably uses a universal notion of compressibility. Indeed, the definition is closely related to the definition of $\seqx$'s monotone \emph{Kolmogorov complexity},
\begin{align}
Km_U(\seqx) = \min\{ |\seqy|: \seqy \in D_U(\seqx)\}, 
\end{align}
the length of $\seqx$'s single shortest description.\footnote{The
	definition of monotone Kolomogorov complexity does not, however, \emph{quite} correspond to that of the Solomonoff-Levin semi-measure, not even in the ``up to constant'' sense of the invariance theorem. See \citep[sect.\ A.3.2]{Ste18phd} for details. } 
	To see that this notion of compressibility is in a certain sense universal, note that if some (long) sequence $\seqx$ has a shorter description (is compressible) via \emph{some} algorithm (via \emph{some} machine $M$), then, up to a constant (the index for $U$ to emulate $M$), it also has a shorter description (is compressible) via $U$.  

Furthermore, since different universal machines can also emulate each other, the choice of particular universal machine is also inconsequential, up to a constant.
	The latter fact is known as the \emph{invariance theorem}, and is taken to provide the notion of Kolmogorov complexity a certain objectivity or robustness.\footnote{Though
	the 	``up to a constant'' is not innocuous. It implies that if we first fix two different machines, then, as we investigate an increasingly long sequence, the two average Kolmogorov complexities (i.e., complexities divided by sequence length) via the two different machines will convergence; in that precise sense, the choice is inconsequential. However, if we switch perspective to the question of the complexity of a given particular sequence, then the difference between the Kolmogorov complexities via two different subsequent choices of machine  can be literally anything (also see \citealp[sect.\ 5.2.2]{Ste18phd}). \label{fn:invar} }  Solomonoff's first proposal of something like Kolmogorov complexity, including an indication of this invariance, makes him one of the three independent founders (with \citealp{Kol65pit,Cha69jacm}) of algorithmic information theory. 
	

On this interpretation that a Solomonoff-Levin semi-measure $\lambda_U$ assigns higher probability to more compressible (in that sense, \emph{simpler}) sequences, we also have that $\lambda_U(x \mid \seqx) = \lambda_U(\seqx x)/\lambda_U(\seqx) $ is higher for the one-bit extension $\seqx x$ of $\seqx$ which is the more compressible (the more simple). Thus we have that the Solomonoff-Levin semi-predictor implements a certain simplicity bias, in line with Occam's razor.

\subsubsection{Some skepticism}\label{ssec:occamskep}
It should first be noted that the invariance of Solomonoff's ``algorithmic probability,'' i.e., between the Solomonoff-Levin semi-measures, is a long way from saying that two Solomonoff-Levin semi-measures make the same predictions.\footnote{Echoing
	footnote \ref{fn:invar}, for any sequence, the ratio between the probabilities assigned to this sequence by the subsequent choice of two different Solomonoff-Levin semi-measures can be literally anything. We at most have that for two fixed choices of machine, as we investigate an increasingly long length-$t$ sequence, the $t$-th roots of the probability assignments converge (see \citealp[sect.\ 3.2.5.6]{Ste18phd}). This does not even imply that on every sequence the \emph{predictions} converge (fn.\ \ref{fn:detconv} above). } 
	The problem of \emph{variance}, or the ``subjectivity'' involved in the choice of universal machine (or equivalently, by theorem \ref{thrm:repr}, the choice of computable weight function for a universal mixture) is widely perceived to be a main challenge for Solomonoff's theory (see \citealp[sect.\ 3.2.5]{Ste18phd}).	
	
Turning to the association with Occam's razor, the question is how robust this simplicity interpretation really is. Theorem \ref{thrm:repr} also shows that the Solomonoff-Levin semi-measures need not be written as universal transformations of the uniform measure. In particular, we could use any continuous measure $\mu$ instead of $\lambda$, losing this clear-cut interpretation of description lengths and compressibility.\footnote{The
	same holds for the association, via the equal uniform probabilities of equal-length descriptions, with the principle of indifference \citep{Sol64ic,RatHut11e}.\label{fn:indiff}}

This huge flexibility in how we can represent the Solomonoff-Levin semi-measures points at something more important. Namely, it is not clear what work the perceived simplicity bias actually does. The distinctive goodness property of the Solomonoff-Levin semi-predictors is their universal reliability. The property of the Solomonoff-Levin semi-predictors which provably makes for their universal reliability is their universal dominance---a property which, of course, does not depend on how exactly we represent them. Thus if the question is what the theory developed here can tell us about what makes for good prediction, then the answer is: universal dominance. It is not clear what Occam's razor adds to this answer.\footnote{A 
	similar point is made in (\citealp{Ste16pos}; \citeyear[sect.\ 5.1]{Ste18phd}).}

This in turn ties in with something more important still. Presumably, the use of inferring methodological principles from the Solomonoff-Levin universal prediction methods is ultimately to give us guidance in the design or understanding of actually implementable or implemented machine learning algorithms. So in particular the use of the association of Occam's razor is presumably to give some justification for implementing a simplicity bias in learning algorithms, or some explanation why learning algorithms implementing a certain simplicity bias work well. But is there really such guidance to be had from the Solomonoff-Levin definition?

\subsubsection{A gold standard}\label{sssec:approxs}
Their formal incomputability disqualifies the Solomonoff-Levin semi-predictors from actual implementation. Nevertheless, so runs the idea in the literature, ``algorithmic probability can serve as a kind of `gold standard' for induction systems'' \citep[p.\ 83]{Sol97jcss}. 

The situation is  not unlike the project of \emph{ideal epistemology} (\citealp{Car22m}, also see \citealp[p.\ 298]{Net23pos}). Here we seek to identify robust norms of rationality, which could never be satisfied by actual, cognitively bounded reasoners, but which could serve as an ideal that can be \emph{approximated} \citep[p.\ 1156]{Car22m}. Similarly, the Solomonoff-Levin definition may serve as an ideal of prediction, and indeed we can design methods that explicitly approximate the Solomonoff-Levin semi-predictors and so this ideal.  The usual idea is that we  simply impose a finite bound on the potentially unbounded quantities in the definition, in particular on the number of computation steps (see, e.g., \citealp{GMGHODCRWMAV24icml}). We would recover the original definition in the limit as we let these quantities grow to infinity, so there is a clear sense in which these finite-bound definitions are (and, for larger bounds, are \emph{better}) approximations.
 
Moreover, this process of getting closer and closer to a gold standard of prediction by scaling up resources might be an appropriate high-level picture of the development of the field of machine learning. Perhaps this picture of approximating the Solomonoff-Levin semi-predictors is a helpful model of understanding what is going on in the field. Perhaps some important developments  were even explicitly motivated by the attempt to approximate the Solomonoff-Levin semi-predictors. And perhaps there is actually normative guidance to be had from the theory, namely that researchers \emph{should} be explicitly moving towards a project of better approximating the Solomonoff-Levin semi-predictors.

\subsubsection{Some skepticism}\label{sssec:approxsskep}
We have here a descriptive and a normative question. To start with the first, is it helpful to model or \emph{describe} machine learning research as being on a path towards the ideal of the Solomonoff-Levin semi-predictors?    

A problem for the descriptive account is that from the perspective that seems relevant here, this idea of approximations is meaningless. Namely, for \emph{any} actual (i.e., computable) prediction method, we could pick a universal Turing machine and finite bounds so that this method would count as an approximation of the Solomonoff-Levin definition.\footnote{Namely, 
	we could define a universal Turing machine such that within this finite approximation the only machine it emulates is one corresponding to this computable method.
	}\textsuperscript{,}\footnote{This 
observation is related to the claim that Bayesianism is empty insofar any conclusion can post-hoc be reconstructed as resulting from a Bayesian inference \citep{Alb01inc}. Of course, this is not the right perspective when we care about the normative question how to proceed in inference. But it is for the current descriptive question.} 
	That suggests that whatever state the field currently is in (in particular, whatever necessarily computable methods have been developed so far), we could model it as in a stage of approximating the Solomonoff-Levin predictors. 
	
That of course leaves open the possibility that prominent machine learning researchers were or are  motivated in their work to approximate the Solomonoff-Levin semi-predictors. Perhaps; this is an empirical or psychological question. The more substantial question is the \emph{normative} question whether they should.

The problem here is that the normative  advice to approximate the Solomonoff-Levin semi-predictors seems exceedingly weak, to the point of not really giving us anything. The immediate question arises, which of the infinitely many different Solomonoff-Levin semi-predictors should one seek to approximate? This is the problem of variance (sect.\ \ref{ssec:occamskep} above). Perhaps there are successful independent arguments for a preferred ``most natural'' universal Turing machine and hence Solomonoff-Levin semi-predictor. But in making such an argument we need to move beyond the theoretical story so far, which justified the Solomonoff-Levin semi-predictors---and \emph{all} equally---as a gold standard by their universal reliability. 
More generally, the theoretical story so far does not give us any guidance on what we presumably would  care about in practice, namely what approximations are in the shorter term better than others. Whatever the computable approximation we would use at any finite stage, it must automatically lose the one distinctive goodness property of the Solomonoff-Levin semi-predictors: 
their universal reliability.\footnote{As
	a consequence, we also have that two different finite approximations do not necessarily $\mu$-a.s.\ converge on the same predictions for a computable data-generating $\mu$. \citet{Net23pos} uses this observation to argue that proponents of Solomonoff induction face a dilemma: they cannot at the same time say that the problem of incomputability is overcome by using a computable approximation, and that the problem of variance 
is overcome by a.s.\ convergence in the limit.}

In the end, just the idea that we should try to approximate the theoretical ideal of universality does not appear to have much more substance than the counsel to include as many hypotheses as possible in our prior.\footnote{And
	weigh those by their simplicity, proponents would add. But even if we accept that a Solomonoff-Levin semi-predictor implements an objective simplicity bias, because it is based on Kolmogorov complexity, any actual predictor must correspond to some computable description length function. Since any predictor corresponds to \emph{some} description length function and so can be seen to implement \emph{some} simplicity bias (fn.\ \ref{fn:codesprobs} above), absent more guidance on what are legitimate compressors, this advice is again meaningless.} 
	That does not seem very deep as far as methodological advice goes (did we need all this math for that?). At the very least, one would like further guidance on a particular way of doing so, analogous to further arguments for a particular approximation or universal Turing machine. But such arguments look suspiciously like arguments for particular inductive biases.  Indeed, instead of trying to approximate universal prediction, it seems more helpful counsel to take seriously the lesson that every (actual, computable) learning method must have an inductive bias, and to think hard about what inevitable such bias your learning algorithm has or should have. 

\section{Conclusion}\label{sec:concl}

\cite{MinReeValLou25nc} provide one recent example of work which brings up Solomonoff induction in the context of seeking a better understanding of generalization in deep learning, and in particular of the inductive biases of deep neural nets.\footnote{The
	context of such work is the ``generalization puzzle,'' prompted by the apparent inability of classical machine learning theory to account for the generalization success of (in particular) deep neural nets \citep{BGKP22inc}.}
	The authors' general idea is to use a Bayesian analysis to identify an inductive bias towards simplicity, and that ``if this inductive bias in the prior matches the simplicity of structured data then it would help explain why DNNs generalize so well'' (\citealp[p.\ 2]{MinReeValLou25nc}; also see \citealp{ValCamLou19arxiv}). 
	
Specifically, the authors study the learning of Boolean functions by a fully connected neural network with ReLU activation functions, and consider a prior $P$ over functions which arises from random sampling of the network's parameter values. This prior, they observe, is insensitive to different choices of sampling distribution, and exhibits a bias ``towards simple functions with low descriptional complexity'' (ibid., p.\ 2). More precisely, the prior probabilities $P(f)$ track the terms $2^{-\tilde{K}(f)}$, where $\tilde{K}$  is a complexity measure based on the Lempel-Ziv (LZ) compression algorithm, which serves as ``a proxy for the true (but uncomputable) Kolmogorov complexity''(ibid.). 
 
Given the emphasis on approximating Kolmogorov complexity, it is no surprise that the authors conclude that an ``interesting direction to explore is the more formal arguments in [algorithmic information theory] relating to the optimality of Solomonoff induction'' (ibid., p.\ 7), and that ``there is an interesting qualitative similarity between Solomonoff induction and learning with simplicity-biased DNNs. In both cases, all hypotheses are included in principle (Epicurus) while simpler hypotheses are exponentially preferred (Occam)'' (ibid., suppl.\ info., p.\ 10).\footnote{The
	reference to Epicurus stems from (\citealp[pp.\ 344f]{LiVit92jcss}; also see \citealp[sect.\ 5.1.1]{LiVit19}), who formulate as ``the principle of Epicurus' multiple explanations (or indifference)'' that ``[i]f more than one theory is consistent with the observations, keep all theories.'' This they further associate with the principle to assign each equal probability (see fn.\ \ref{fn:indiff}).}

Cautiously formulated as these hopes for further insights from Solomonoff induction are,\footnote{Or
	 maybe this is just  ``decorative prose'' \citep{KeiLou25philsci}\dots} 
	 our  discussion suggests that we have reason to be skeptical. The finding that (certain) DNNs can be seen to operate (under certain circumstances) under a prior that resembles the distribution corresponding to a popular compression algorithm,\footnote{Or
	 indeed class of such algorithms: the authors write, with reference to \citep{ValCamLou19arxiv,BPKB23acl}, that ``[o]ther complexity measures give similar results'' \citep[p.\ 2]{MinReeValLou25nc}.} and that this prior generally matches real-world data, is, of course, interesting and possibly significant. The use of coding theory may also be helpful to analyze this  further.\footnote{In
	 particular the relationship between (via any given compressor) descriptional complexity and the number of objects with this complexity \citep[p.\ 4]{MinReeValLou25nc}.} 
	 But the suggestion that this particular prior is a proxy of ``the'' universal simplicity prior, namely Solomonoff's (with its link to ``optimality'' and Occam's razor), does not seem very useful, or even meaningful. 
Rather than thinking that Solomonoff induction might somehow step in here to complete  an explanation of generalization success, it seems more accurate and useful to recognize that the identification of this \emph{particular} inductive bias, aligning with this \emph{particular} type of compression algorithm (and, if you will, this \emph{particular} type of simplicity bias), can only be a starting point. 

Namely, for any actual, computable, learning procedure, the picture we started this chapter with is left intact. Every learning method must have restrictive inductive biases, restrictive inductive assumptions. A Bayesian analysis helps us to locate these restrictive assumptions, namely in the prior. The Bayesian philosophical point that this prior is strictly unjustified translates into the practical point that we should strive to understand the inductive assumptions involved. In what kind of learning situations, for what kind of data are these assumptions appropriate---and in what cases are they not? The  lure of a theory of universal induction notwithstanding, this is no different for the inductive biases in deep neural networks.

\end{document}